\documentclass[runningheads]{llncs}
\usepackage[T1]{fontenc}
\usepackage{amsmath, amssymb}
\usepackage{mathtools}
\usepackage{comment}

\usepackage{tikz}
\usetikzlibrary{fit,backgrounds, arrows.meta, calc, positioning, matrix}

\usepackage{algorithm}
\usepackage{algpseudocode}
\usepackage{placeins} % in preamble
\usepackage{url}
\newcommand{\doi}[1]{\url{https://doi.org/#1}}
\newcommand{\asc}{\operatorname{asc}}

\spnewtheorem{fact}[theorem]{Fact}{\bfseries}{\itshape}

\providecommand{\qedsymbol}{\ensuremath{\square}}
\newcommand{\qedbox}{\hfill\qedsymbol}

\title{A Unary-to-Nonunary Transition in the Accepting-State Spectrum of Right Quotient for Permutation Automata}
\titlerunning{Beyond Unary Right Quotient for Permutation Automata}
\usepackage{bbding}

\author{Samuel German \Envelope}
\authorrunning{Samuel German}
\institute{University of California, San Diego, USA \\
\email{sgerman@ucsd.edu}}

\begin{document}
\maketitle
\begin{abstract}
This paper determines the accepting-state spectrum of the right-quotient operation for permutation automata, solving an open problem posed by Rauch and Holzer. They showed that in the unary case in which the languages whose right quotient we consider are languages over an alphabet with but one symbol, the accepting state complexities which could ensue for automata recognizing their quotient is exactly all integers between $1$ and $mn$ inclusive when $m,n \ge 1$; we show that for arbitrary alphabets the accepting-state spectrum is

\[
g^{\asc}_{-1,\mathrm{PFA}}(m,n)=
\begin{cases}
\{0\}, & \text{if } m=0 \text{ or } n=0,\\[1mm]
\mathbb{N}_{>0}, & \text{if } m,n\ge 1.
\end{cases}
\]

Our result shows that every positive accepting-state complexity ensues for right quotient as long as both of the languages of which we consider the right quotient are nonempty. $0$ becomes the only unavoidable magic value.  Moreover, we show that an alphabet with just two letters can witness any arbitrary $\alpha \ge 1$ when $m,n \ge 1$, in that languages can be defined over a binary alphabet with accepting-state complexity $m$ and $n$ respectively for which their right quotient has accepting-state complexity $\alpha$.

 The first part our proof is that we show that if $m,n\ge 1$, then 
language $KL^{-1}$ cannot be empty when $K$ and $L$ are accepted by minimal
permutation automata with $m$ and $n$ final states respectively. We get this from how the transition action is bijective. The other important part is that we construct a ternary witness pair $(A^{\mathrm q}_{m,\alpha},B^{\mathrm q}_{n,\alpha})$ such that
$
\asc(L(A^{\mathrm q}_{m,\alpha}))=m$, $
\asc(L(B^{\mathrm q}_{n,\alpha}))=n$, and $
\asc\!\bigl(L(A^{\mathrm q}_{m,\alpha})L(B^{\mathrm q}_{n,\alpha})^{-1}\bigr)=\alpha $; we do this for every $m,n\ge 1$ and every $\alpha \ge m$. We eventually simulate the ternary witness pair over just a binary alphabet to realize every $\alpha > mn$.

What happens is that the words accepted by
$B^{\mathrm q}_{n,\alpha}$ induce a point stabilizer in a symmetric
group, and the standard quotient construction giving an automaton recognizing the quotient language then saturates the original final
set of $A^{\mathrm q}_{m,\alpha}$ to a full orbit; this gives a minimal quotient
automaton with exactly $\alpha$ final states. The complete spectrum which is our result follows from this, solving the open question of what the right quotient's accepting-state complexity for permutation automata is. 

\end{abstract}
\section{Introduction}
\label{sec:introduction}
Herein we completely solve the open problem posed by Rauch and Holzer to give the accepting state spectrum of the right quotient operation for permutation automata in the case of arbitrary alphabets~\cite{RauchHolzer2023}. Those authors showed that said spectrum is $[1,mn]$ in the unary case $m,n\ge 1$ but left as an open problem the accepting-state spectrum when the alphabet is arbitrary. We show that there is a transition here from transitioning to the unary to non-unary alphabet case, in that the spectrum goes from being finitely bounded as $[1,mn]$ to being complete in the sense that every positive integer then appears in the spectrum. In particular, we show that the binary-alphabet case already witnesses the complete spectrum of all positive integers, so long as each of $m$ and $n$ are at least 1 (Corollary~\ref{cor:quotient-binary-spectrum}). 

Over arbitrary alphabets, our result is that the accepting-state spectrum is
\[
g^{\asc}_{-1,\mathrm{PFA}}(m,n)=
\begin{cases}
\{0\}, & \text{if } m=0 \text{ or } n=0,\\[1mm]
\mathbb{N}_{>0}, & \text{if } m,n\ge 1.
\end{cases}
\]

Accepting-state complexity is a concept which was introduced and studied by Dassow~\cite{Dassow2016}. It measures the number of accepting states in the minimal automata (automata with the fewest states) recognizing a language. It is already known that when we do not restrict ourselves to permutation automata, the accepting state-spectrum for the right quotient operation is $\{0\}\cup\mathbb{N}$ for $m,n \ge 1$~\cite{HospodarHolzer2020Ranges}; the open question we solve is what happens in the restricted case of permutation automata. 

Accepting-state complexity in the case of permutation automata underwent systematic study by Rauch and Holzer~\cite{RauchHolzer2023}. They focus on permutation automata due to its relevance in, for example, decomposing general automata, and study accepting-state complexity for permutation automata of the right quotient in particular as a way to better understand this measure, due to its divergence in behavior from other descriptional complexity measures such as state complexity~\cite{RauchHolzer2023}. As discussed, for the right quotient here they show the unary case but explicitly leave the case with general alphabet open.

Our proof can be broken up into two main parts. First, we show that if
$m,n\ge 1$, then $\asc(KL^{-1})\neq 0$ for all languages $K$ and $L$
accepted by permutation automata with $\asc(K)=m$ and $\asc(L)=n$. The bijectivity that comes from being a permutation is integral here. Second, for every $m,n\ge 1$ and every
$\alpha\ge m$, we construct a ternary witness pair
$(A^{\mathrm q}_{m,\alpha},B^{\mathrm q}_{n,\alpha})$ such that
$
\asc(L(A^{\mathrm q}_{m,\alpha}))=m$, $
\asc(L(B^{\mathrm q}_{n,\alpha}))=n$, and $
\asc\!\bigl(L(A^{\mathrm q}_{m,\alpha})L(B^{\mathrm q}_{n,\alpha})^{-1}\bigr)=\alpha$. The full spectrum of all positive integers comes from this and the known $[1,mn]$ from the unary case. 

The construction we use to give an automaton with having $\alpha$ final states for any integer $\alpha \ge m$ uses some basic group theory. We give an automaton  $A^{\mathrm q}_{m,\alpha}$ acting on  $[k]$ for $k=\alpha+1$, and another automaton  $B^{\mathrm q}_{n,\alpha}$ which acts on $[k]\times\mathbb{Z}_{n+1}$. We carefully design the language accepted by $B^{\mathrm q}_{n,\alpha}$ to be such that it induces permutations on the first coordinate that are exactly the point stabilizer
$\operatorname{Stab}_{S_k}(k)$. Then, the standard quotient construction which deals with $A^{\mathrm q}_{m,\alpha}$ and $B^{\mathrm q}_{n,\alpha}$ as we have defined them here gives an automaton with $\alpha$ final states, which comes from how, as we show, the final set of  $A^{\mathrm q}_{m,\alpha}$ is saturated to the whole orbit$[k-1]$. 

We show that the automata $A^{\mathrm q}_{m,\alpha}$ and $B^{\mathrm q}_{n,\alpha}$ we gave are minimal, so, given what we have stated above, the languages they recognize have accepting state complexities of $m$ and $n$ respectively, and we show minimality of the quotient, so that we can say, based on what is mentioned of it above, that the language which it recognizes has accepting state complexity $\alpha$.

\section{Preliminaries}
\label{sec:preliminaries}

Herein we only consider deterministic, complete automata. Another one of our conventions is that we operate in accordance with the standard meaning $[n] := \{1,\dots,n\}$.

A \emph{deterministic finite automaton} (DFA) is a tuple
$A = (Q,\Sigma,\cdot,s,F)$, where $Q$ is a finite state set, $\Sigma$ is a
finite input alphabet, $\cdot : Q \times \Sigma \to Q$ is the transition
function, $s \in Q$ is the initial state, and $F \subseteq Q$ is the set of
final states. The language accepted by $A$ (we also say recognized by $A$) is
\[L(A) := \{\, w \in \Sigma^{*} \mid s \cdot w \in F \,\}.\]
We say that a DFA is \emph{minimal} if there doesn't exist any other DFA which recognizes the same language as it but which has fewer states. 

A \emph{permutation automaton} is a DFA in which, for every $a \in \Sigma$, the map
$q \mapsto q \cdot a$ is a permutation of $Q$. Hence every word
$w \in \Sigma^{*}$ induces a permutation $\pi_{w}$ of $Q$, given by
$q \pi_{w} = q \cdot w$ for $q \in Q$.

We extend the transition function $\cdot$ to words by setting $q\cdot\varepsilon=q$ and $q\cdot(wa)=(q\cdot w)\cdot a$ for $q\in Q$, $w\in\Sigma^*$, and $a\in\Sigma$. Since each letter acts as a permutation of $Q$, each word $w\in\Sigma^*$ induces a permutation $\pi_w$ of $Q$, given by $q\pi_w=q\cdot w$ for all $q\in Q$.

For a regular language $L$, the \emph{accepting-state complexity} of $L$,
denoted by $\asc(L)$, is the number of final states in the DFA accepting $L$ which has the smallest number of final states out of any DFA which accepts $L$.

For languages $K,L \subseteq \Sigma^{*}$, the \emph{right quotient} of $K$
by $L$ is \[K L^{-1} := \{\, x \in \Sigma^{*} \mid \exists y \in L:\ xy \in K \,\}.\]

\[
g^{\asc}_{-1,\mathrm{PFA}}(m,n)
:= \left\{
\alpha \ge 0 \;\middle|\;
\begin{array}{@{}l@{}}
\exists \Sigma,\ \exists K,L \subseteq \Sigma^*,\\
K,L \text{ accepted by permutation automata over } \Sigma,\\
\asc(K)=m,\ \asc(L)=n,\ \asc(KL^{-1})=\alpha
\end{array}
\right\}.
\]
The following standard fact, due to Dassow~\cite{Dassow2016}, allows us to
identify accepting-state complexity directly from a minimal witness automaton.

\begin{proposition}[{\cite{Dassow2016}}]
\label{prop:minimal-realizes-asc}
Let $A=(Q,\Sigma,\cdot,s,F)$ be a minimal DFA. Then
\[
\asc(L(A)) = |F|.
\]
\end{proposition}

The two lemmas below are standard, and are used by us throughout the paper. 

\begin{lemma}[{cf.~\cite[Section~4, p.~407]{LinggOliveiraWolf2022}}]
\label{lem:accessible-part-permutation}
The accessible part of a permutation automaton is again a
permutation automaton.
\end{lemma}

\begin{lemma}[{cf.~\cite[Chapter~II, Proposition~3.13]{Pin2025}}]
\label{lem:positive-words-generate-group}
Let $X$ be a finite set and let $c_1,\dots,c_r\in \mathrm{Sym}(X)$. Then
the submonoid of $\mathrm{Sym}(X)$ generated by $c_1,\dots,c_r$ is equal
to the subgroup $\langle c_1,\dots,c_r\rangle$. Equivalently, every
element of $\langle c_1,\dots,c_r\rangle$ is induced by some positive word
over $c_1,\dots,c_r$.
\end{lemma}

\begin{lemma}
\label{lem:quotient-final-set}
Let $A=(Q,\Sigma,\cdot,s,F)$ and $B$ be DFAs  which are over the same alphabet. Define
\[
\widetilde{F} := \{\, q \in Q \mid \exists w \in L(B): q \cdot w \in F \,\},
\]
and let $A/B := (Q,\Sigma,\cdot,s,\widetilde{F})$. Then
\[
L(A/B) = L(A)L(B)^{-1}.
\]
If $A$ is a permutation automaton and $\pi_{w}$ denotes the permutation of
$Q$ induced by $w$, then
\[
\widetilde{F}
   = \bigcup_{w \in L(B)} \pi_{w}^{-1}(F).
\]
In particular, if $G_{B} := \{\, \pi_{w} \mid w \in L(B) \,\}$, then
\[
\widetilde{F}
   = \bigcup_{\pi \in G_{B}} \pi^{-1}(F).
\]
\end{lemma}

\begin{proof}
For every $x\in\Sigma^{*}$, we have
\begin{align*}
x\in L(A/B)
 &\iff s\cdot x\in \widetilde{F} \\
 &\iff \exists w\in L(B):\ (s\cdot x)\cdot w\in F \\
 &\iff \exists w\in L(B):\ s\cdot xw\in F \\
 &\iff \exists w\in L(B):\ xw\in L(A) \\
 &\iff x\in L(A)L(B)^{-1}.
\end{align*}
Thus $L(A/B)=L(A)L(B)^{-1}$.

Now assume that $A$ is a permutation automaton. Then for every
$q\in Q$ and $w\in\Sigma^{*}$, we have $q\cdot w=q\pi_{w}$. Hence
\begin{align*}
q\in\widetilde{F}
 &\iff \exists w\in L(B):\ q\cdot w\in F \\
 &\iff \exists w\in L(B):\ q\pi_{w}\in F \\
 &\iff q\in \bigcup_{w\in L(B)} \pi_{w}^{-1}(F).
\end{align*}
Therefore
\[
\widetilde{F}=\bigcup_{w\in L(B)} \pi_{w}^{-1}(F).
\]
Replacing the indexing set of words by the set $G_{B}$ of induced
permutations does not change this union, so also
\[
\widetilde{F}=\bigcup_{\pi\in G_{B}} \pi^{-1}(F).
\]\qedbox
\end{proof}

We also use standard group-action notation. For $n \ge 1$, let $S_{n}$
denote the symmetric group on $[n]$. If $G \le S_{n}$, then $G$ acts on
$[n]$ by $i\pi = \pi(i)$ for $i \in [n]$ and $\pi \in G$, and we write
\[iG := \{\, i\pi \mid \pi \in G \,\}\] for the corresponding orbit. We also
write \[\operatorname{Stab}_{G}(i) := \{\, \pi \in G \mid i\pi = i \,\}\] for the
point stabilizer of $i$. We treat the letters of the alphabet of the automata as functions, composing them.  Thus, if
$w=c_{1}c_{2}\cdots c_{r}$, then the permutation induced by $w$ is
$\pi_{w}=c_{r}\circ \cdots \circ c_{2}\circ c_{1}$, so that
$q\pi_{w}=q\cdot w$ for every state $q$.

We also use the standard facts that, for $n\ge 2$, the cycle
$(1\,2\,\dots\,n)$ together with the transposition $(1\,2)$ generates
$S_n$, that the natural action of $S_n$ on $[n]$ is transitive, and
that the point stabilizer $\operatorname{Stab}_{S_k}(k)$ acts transitively
on $[k-1]$; see, for example,
Dixon--Mortimer~\cite[Chapter~1]{DixonMortimer1996}. 

\section{Exact Spectrum of Right Quotient}
\label{sec:quotient}
Now we begin to directly turn attention to determining the exact accepting-state spectrum for right quotient over
permutation automata. Our first task is to show that when both of the languages whose quotient we consider are nonempty, the accepting-state complexity of the quotient is positive.

\begin{proposition}
\label{prop:quotient-zero-impossible}
For all $m,n \ge 1$, we have $0 \notin g^{\asc}_{-1,\mathrm{PFA}}(m,n)$.
\end{proposition}

\begin{proof}
Let $K,L \subseteq \Sigma^{*}$ be languages accepted by permutation automata
and satisfying $\asc(K)=m\ge 1$ and $\asc(L)=n\ge 1$. Then
$K \neq \emptyset$ and $L \neq \emptyset$. Let
$A=(Q,\Sigma,\cdot,s,F)$ be the accessible part of a permutation automaton
accepting $K$. By
Lemma~\ref{lem:accessible-part-permutation}, $A$ is again a permutation
automaton. Let $B$ be a permutation automaton accepting $L$.

Let $f \in F$, we are able to do because $F \neq \emptyset$ as a result of $K$ being nonempty. As $L\neq\emptyset$, let $w$ be such that $w \in L(B)$. Now we have that $q \mapsto q \cdot w$ is a bijection from $Q$ to $Q$ due to how $A$ is a permutation automaton. So there must exist a state $q \in Q$ for which $q \cdot w = f$. $q \in \widetilde{F}$: by Lemma~\ref{lem:quotient-final-set}, the accepting states of the quotient
automaton $A/B$ is
\[\widetilde{F} = \{\, r \in Q \mid r \cdot u \in F \text{ for some } u \in L(B) \,\}.\] Thus we have that $\widetilde{F}\neq\emptyset$.

Furthermore, $x$ is accepted by $A/B$ because, as  $A$ is accessible, there is some $x \in \Sigma^{*}$ such that
$s \cdot x = q$. By
Lemma~\ref{lem:quotient-final-set}, this means
$x \in L(A)L(B)^{-1} = KL^{-1}$. Hence $KL^{-1}\neq\emptyset$. Since
the empty language is the only language of accepting-state complexity $0$, it
follows that $\asc(KL^{-1}) \neq 0$. Therefore
$0 \notin g^{\asc}_{-1,\mathrm{PFA}}(m,n)$. \qedbox
\end{proof}

\begin{proposition}[{\cite[Corollary~3.30]{RauchHolzer2023}}]
\label{prop:quotient-unary-range}
For all $m,n \ge 1$ and all $\alpha \in \{1,\dots,mn\}$, we have
\[
\alpha \in g^{\asc}_{-1,\mathrm{PFA}}(m,n).
\]
\end{proposition}

\begin{proof}
By \cite[Corollary~3.30]{RauchHolzer2023}, every value in
$\{1,\dots,mn\}$ is realized already by unary permutation automata. These
unary witnesses also count when we are in the context of allowing any arbitrary alphabet. \qedbox
\end{proof}

Since $n\ge 1$, we have $mn \ge m$. Consequently, once we realize every
$\alpha \ge m$, the low range $\{1,\dots,mn\}$ and the high range
$\{\alpha \in \mathbb{N} \mid \alpha \ge m\}$ together cover all positive
integers.

We therefore fix integers $m,n \ge 1$ and $\alpha \ge m$, and set
$k := \alpha+1$. We now define two minimal automata of accepting state complexity $m$ and $n$ respectively for which the quotient construction applied to them produces a minimal automata with $\alpha$ final states. 

Let the alphabet be $\Sigma := \{a,b,c\}$. On $[k]$, let
$a=(1\,2\,\dots\,k)$, $b=(1\,2)$, and $c=\mathrm{id}_{[k]}$.
Since $a$ and $b$ generate $S_k$~\cite[Chapter~1]{DixonMortimer1996},
these letters induce a permutation action on $[k]$.

Our first definition here is 
\[
A^{\mathrm{q}}_{m,\alpha}
   := \bigl([k],\Sigma,\cdot_A,1,F_A\bigr),\] wherein \[F_A := \{1,\dots,m\},
\] and, for $p \in [k]$, $
p\cdot_A a = pa,
p\cdot_A b = pb,
p\cdot_A c = p$. 
Because $\alpha \ge m$, we have $m \le k-1$, and hence
$F_A \subseteq [k-1]$.
In particular, the distinguished point $k$ is nonfinal in
$A^{\mathrm{q}}_{m,\alpha}$.

For the second automata which will be one of our witnesses, let $\mathbb{Z}_{n+1}$ be written additively.
Define
\[
B^{\mathrm{q}}_{n,\alpha}
   := \bigl([k]\times \mathbb{Z}_{n+1},\Sigma,\cdot_B,(k,0),F_B\bigr),
\]
in which 
\[
F_B := \{\, (k,i) \mid i \in \mathbb{Z}_{n+1}\setminus\{n\} \,\},
\]
and the transition function is given by
$
(p,i)\cdot_B a = (pa,i)$,
$(p,i)\cdot_B b = (pb,i)$, and $
(p,i)\cdot_B c = (p,i+1)
$
for all $(p,i)\in [k]\times \mathbb{Z}_{n+1}$, where the addition in the
second coordinate is modulo $n+1$.

The automaton $B^{\mathrm{q}}_{n,\alpha}$ is designed so that the
permutations of $[k]$ induced by words of $L(B^{\mathrm{q}}_{n,\alpha})$
are exactly the permutations fixing $k$.

The rest of this section is devoted to showing that
\[
\asc(L(A^{\mathrm{q}}_{m,\alpha})) = m,\] \[\asc(L(B^{\mathrm{q}}_{n,\alpha})) = n,\] and \[\asc\!\bigl(L(A^{\mathrm{q}}_{m,\alpha})L(B^{\mathrm{q}}_{n,\alpha})^{-1}\bigr)
= \alpha.
\]

We first verify that the witness automata $A^{\mathrm{q}}_{m,\alpha}$ and  $B^{\mathrm{q}}_{n,\alpha}$ themselves are minimal.

\begin{lemma}
\label{lem:quotient-witnesses-minimal}
The automata $A^{\mathrm{q}}_{m,\alpha}$ and $B^{\mathrm{q}}_{n,\alpha}$ are
minimal permutation automata, and 
$
|F_A|=m$ and $|F_B|=n
$.
\end{lemma}

\begin{proof}
Each letter induces a bijection through the transition function from the state set of $A^{\mathrm{q}}_{m,\alpha}$ to the state set of $A^{\mathrm{q}}_{m,\alpha}$:
the letters $a$ and $b$ are permutations of $[k]$, and $c$ acts as
the identity function. Hence $A^{\mathrm{q}}_{m,\alpha}$ is a permutation automaton.
Likewise, each letter induces a bijection when the transition function fixes it from the state set 
$[k]\times \mathbb{Z}_{n+1}$ of $B^{\mathrm{q}}_{n,\alpha}$ to said state set: the letters
$a$ and $b$ permute the first coordinate and fix the second, while $c$
fixes the first coordinate and acts as the cycle
$i \mapsto i+1 \pmod{n+1}$ on the second. Thus
$B^{\mathrm{q}}_{n,\alpha}$ is also a permutation automaton.

By Lemma~\ref{lem:positive-words-generate-group}, every permutation in
$\langle a,b\rangle=S_k$ is induced by some word in $\{a,b\}^{*}$.

For $A^{\mathrm{q}}_{m,\alpha}$, we clearly have
$|F_A| = |\{1,\dots,m\}| = m$. Since the natural action of $S_k$ on
$[k]$ is transitive,
every state is reachable from the initial state $1$. To show
distinguishability, let $p,q \in [k]$ with $p \neq q$. Because
$\langle a,b\rangle = S_k$, there exists a permutation $\sigma \in S_k$
such that $p\sigma = 1$ and $q\sigma = k$. Let
$w \in \{a,b\}^{*}$ induce $\sigma$. Then
$p \cdot_A w = 1 \in F_A$ and $q \cdot_A w = k \notin F_A$,
since $F_A \subseteq [k-1]$. Hence $p$ and $q$ are distinguishable, and
therefore $A^{\mathrm{q}}_{m,\alpha}$ is minimal.

For $B^{\mathrm{q}}_{n,\alpha}$, the set of final states is
$F_B=\{\, (k,i) \mid i \in \mathbb{Z}_{n+1}\setminus\{n\} \,\}$, so
$|F_B| = n$. Every state $(p,i)\in [k]\times \mathbb{Z}_{n+1}$ is
reachable from
$(k,0)$: choose a word $x \in \{a,b\}^{*}$ sending $k$ to $p$, and then
append $c^{r}$, where $0 \le r \le n$ is chosen so that
$r \equiv i \pmod{n+1}$. Then $(k,0)\cdot_B xc^{r} = (p,i)$.

It remains to distinguish distinct states of $B^{\mathrm{q}}_{n,\alpha}$.
Let $(p,i)\neq(q,j)$.

If $p \neq q$, choose $\sigma \in S_k$ such that $p\sigma = k$ and
$q\sigma \neq k$, and let $x \in \{a,b\}^{*}$ induce $\sigma$. If
$i \neq n$, then $(p,i)\cdot_B x = (k,i)\in F_B$ and
$(q,j)\cdot_B x = (q\sigma,j)\notin F_B$,
because the first coordinate of the second state is not $k$. If $i=n$, then
$(p,i)\cdot_B xc = (k,0)\in F_B$ and
$(q,j)\cdot_B xc = (q\sigma,j+1)\notin F_B$,
again because $q\sigma \neq k$. Thus the two states are distinguishable.

If $p=q$, then necessarily $i\neq j$. Choose $\tau \in S_k$ with
$p\tau = k$, and let $x \in \{a,b\}^{*}$ induce $\tau$. Choose an integer
$r \in \{0,\dots,n\}$ such that $j+r \equiv n \pmod{n+1}$. Then
$(q,j)\cdot_B xc^{r} = (k,n)\notin F_B$. Since $i\neq j$, we have
$i+r \not\equiv n \pmod{n+1}$, and therefore
$(p,i)\cdot_B xc^{r} = (k,i+r)\in F_B$.
So these two states are also distinguishable.

Hence all states of $B^{\mathrm{q}}_{n,\alpha}$ are reachable and pairwise
distinguishable, and $B^{\mathrm{q}}_{n,\alpha}$ is minimal.

The final assertions now follow from
Proposition~\ref{prop:minimal-realizes-asc}. \qedbox
\end{proof}

For each word $w\in\Sigma^{*}$, let $\rho_w\in S_k$ denote the
permutation induced by $w$ on $[k]$. Equivalently, $\rho_w$ is the action
of $w$ in $A^{\mathrm{q}}_{m,\alpha}$, and also the action of $w$ on the
first coordinate of $B^{\mathrm{q}}_{n,\alpha}$. Define
\[G_{n,\alpha} := \{\, \rho_w \mid w \in L(B^{\mathrm{q}}_{n,\alpha}) \,\}.\]

The below lemma will prove to be important in the overall construction.

\begin{lemma}
\label{lem:quotient-accepted-stabilizer}
We have \[G_{n,\alpha} = \operatorname{Stab}_{S_k}(k).\]
\end{lemma}

\begin{proof}
Let $w \in L(B^{\mathrm{q}}_{n,\alpha})$. Since $w$ is accepted starting from the
initial state $(k,0)$, we have that $(k,0)\cdot_B w = (k,i)$ for some
$i \in \mathbb{Z}_{n+1}\setminus\{n\}$. But the first coordinate after reading $w$ is precisely
$k\rho_w$. Hence $k\rho_w = k$, so
$\rho_w \in \operatorname{Stab}_{S_k}(k)$. This proves
$G_{n,\alpha} \subseteq \operatorname{Stab}_{S_k}(k)$.

Conversely, let $h \in \operatorname{Stab}_{S_k}(k)$. By
Lemma~\ref{lem:positive-words-generate-group}, there exists a word $x \in \{a,b\}^{*}$ such that $\rho_x = h$.
Because $x$ doesn't have $c$ as any of its occurrences, the second coordinate is kept as
$0$, and therefore $(k,0)\cdot_B x = (kh,0) = (k,0)\in F_B$. Thus
$x \in L(B^{\mathrm{q}}_{n,\alpha})$, so $h=\rho_x \in G_{n,\alpha}$. Hence
$\operatorname{Stab}_{S_k}(k) \subseteq G_{n,\alpha}$.

Therefore $G_{n,\alpha} = \operatorname{Stab}_{S_k}(k)$. \qedbox
\end{proof}

Let $A^{\mathrm{q}}_{m,\alpha}/B^{\mathrm{q}}_{n,\alpha}
   = \bigl([k],\Sigma,\cdot_A,1,\widetilde{F}\bigr)$ be the quotient automaton
from Lemma~\ref{lem:quotient-final-set}.

\begin{lemma}
\label{lem:quotient-computed-final-set}
The final set of the quotient automaton is
\[
\widetilde{F}
 = \bigcup_{h\in G_{n,\alpha}} h^{-1}(F_A)
 = \bigcup_{h\in \operatorname{Stab}_{S_k}(k)} h^{-1}(F_A)
 = \{1,\dots,k-1\}.
\]
\end{lemma}

\begin{proof}
By Lemma~\ref{lem:quotient-final-set},
\[
\widetilde{F}
   = \bigcup_{w\in L(B^{\mathrm{q}}_{n,\alpha})} \rho_w^{-1}(F_A)
   = \bigcup_{h\in G_{n,\alpha}} h^{-1}(F_A).
\]
By Lemma~\ref{lem:quotient-accepted-stabilizer},
\[
\widetilde{F}
   = \bigcup_{h\in \operatorname{Stab}_{S_k}(k)} h^{-1}(F_A).
\]

We now compute this union. Since every $h\in \operatorname{Stab}_{S_k}(k)$
fixes $k$, and since $k \notin F_A$, we have
$k \notin h^{-1}(F_A)$ for
$h\in \operatorname{Stab}_{S_k}(k)$. Therefore
$\widetilde{F}\subseteq \{1,\dots,k-1\}$.

For the reverse inclusion, let $x \in \{1,\dots,k-1\}$. Since $m\ge 1$, the
set $F_A=\{1,\dots,m\}$ is nonempty; choose any $y \in F_A$. The point
stabilizer $\operatorname{Stab}_{S_k}(k)$ acts transitively on
$\{1,\dots,k-1\}$, so there exists
$h \in \operatorname{Stab}_{S_k}(k)$ such that $xh = y$. Hence
$x \in h^{-1}(F_A) \subseteq \widetilde{F}$. Thus
$\{1,\dots,k-1\}\subseteq \widetilde{F}$.

We conclude that $\widetilde{F}=\{1,\dots,k-1\}$. \qedbox
\end{proof}

\begin{lemma}
\label{lem:quotient-result-minimal}
The quotient automaton $A^{\mathrm{q}}_{m,\alpha}/B^{\mathrm{q}}_{n,\alpha}$
is minimal. In particular,
\[
\asc\!\bigl(L(A^{\mathrm{q}}_{m,\alpha})L(B^{\mathrm{q}}_{n,\alpha})^{-1}\bigr)
   = k-1 = \alpha.
\]
\end{lemma}

\begin{proof}
By Lemma~\ref{lem:quotient-computed-final-set}, the quotient automaton is
\[
A^{\mathrm{q}}_{m,\alpha}/B^{\mathrm{q}}_{n,\alpha}
   = \bigl([k],\Sigma,\cdot_A,1,\{1,\dots,k-1\}\bigr).
\]

By Lemma~\ref{lem:positive-words-generate-group}, every permutation in
$\langle a,b\rangle=S_k$ is induced by some word in $\{a,b\}^{*}$. Since
the natural action of $S_k$ on $[k]$ is transitive,
every state is reachable from the initial state $1$.

Now let $p,q \in [k]$ with $p\neq q$. Choose $\sigma \in S_k$ such that
$p\sigma = 1$ and $q\sigma = k$. Let $w \in \{a,b\}^{*}$ induce
$\sigma$. Then $p\cdot_A w = 1 \in \{1,\dots,k-1\}$ and
$q\cdot_A w = k \notin \{1,\dots,k-1\}$. As a result,  $p$ and $q$ are distinguishable, implying that the quotient automaton is
minimal.

By Lemma~\ref{lem:quotient-final-set}, this automaton accepts
$L(A^{\mathrm{q}}_{m,\alpha})L(B^{\mathrm{q}}_{n,\alpha})^{-1}$. Since its final set
has size $k-1$, Proposition~\ref{prop:minimal-realizes-asc} gives
$\asc\!\bigl(L(A^{\mathrm{q}}_{m,\alpha})L(B^{\mathrm{q}}_{n,\alpha})^{-1}\bigr)
   = k-1 = \alpha$. \qedbox

\end{proof}

\begin{theorem}
\label{thm:quotient-spectrum}
For all $m,n \ge 0$,
\[
g^{\asc}_{-1,\mathrm{PFA}}(m,n)=
\begin{cases}
\{0\}, & \text{if } m=0 \text{ or } n=0,\\[1mm]
\mathbb{N}_{>0}, & \text{if } m,n \ge 1.
\end{cases}
\]

\end{theorem}

\begin{proof}
If $m=0$ or $n=0$, then any languages $K,L$ with $\asc(K)=m$ and
$\asc(L)=n$ satisfy $K=\emptyset$ or $L=\emptyset$, and therefore
$KL^{-1}=\emptyset$. Hence the only value that could possibly be in the accepting-state spectrum is $0$. We now show that $0$ is actually attained therein. 

Fix $t\ge 0$. Let $U_t$ be the unary language accepted by the following
permutation automaton over $\{a\}$: if $t=0$, take the one-state automaton
with no final states; if $t\ge 1$, take
\[
C_t:=\bigl([t+1],\{a\},\cdot,1,\{1,\dots,t\}\bigr),
\qquad
p\cdot a \equiv p+1 \pmod{t+1}.
\]
For $t\ge 1$, the automaton $C_t$ is a permutation automaton with exactly
$t$ final states. It is minimal because, for each state $p\in [t+1]$, the
unique word of length at most $t$ that sends $p$ to the nonfinal state $t+1$ is
$a^{d_p}$, where $d_p\in\{0,\dots,t\}$ is determined by
$p+d_p\equiv t+1 \pmod{t+1}$. Thus distinct states have distinct right
languages, and Proposition~\ref{prop:minimal-realizes-asc} gives
$\asc(L(C_t))=t$. The case $t=0$ is taken care of by the automaton consisting of only one nonfinal state, which will accept the empty language.

Now, if $m=0$, choose $K=\emptyset$ over $\{a\}$ and choose $L=U_n$;
then $\asc(K)=0$, $\asc(L)=n$, and $KL^{-1}=\emptyset$. Similarly, if
$n=0$, choose $L=\emptyset$ and $K=U_m$. Hence
$0\in g^{\asc}_{-1,\mathrm{PFA}}(m,n)$ whenever $m=0$ or $n=0$. This
proves $g^{\asc}_{-1,\mathrm{PFA}}(m,n)=\{0\}$ if $m=0$ or $n=0$.

Now assume $m,n \ge 1$. By Proposition~\ref{prop:quotient-zero-impossible},
we have $0 \notin g^{\asc}_{-1,\mathrm{PFA}}(m,n)$. By
Proposition~\ref{prop:quotient-unary-range}, every
$\alpha \in \{1,\dots,mn\}$ belongs to
$g^{\asc}_{-1,\mathrm{PFA}}(m,n)$. On the other hand, for every
$\alpha \ge m$, Lemma~\ref{lem:quotient-witnesses-minimal} and
Lemma~\ref{lem:quotient-result-minimal} show that the witnesses
$A^{\mathrm{q}}_{m,\alpha}$ and $B^{\mathrm{q}}_{n,\alpha}$ satisfy
$
\asc(L(A^{\mathrm{q}}_{m,\alpha}))=m$,$
\asc(L(B^{\mathrm{q}}_{n,\alpha}))=n$,
and
$
\asc\!\bigl(L(A^{\mathrm{q}}_{m,\alpha})L(B^{\mathrm{q}}_{n,\alpha})^{-1}\bigr)
=\alpha,
$
so $\alpha \in g^{\asc}_{-1,\mathrm{PFA}}(m,n)$. Since
$mn \ge m$, the sets $\{1,\dots,mn\}$ and
$\{\alpha \in \mathbb{N} \mid \alpha \ge m\}$ together equal
$\mathbb{N}_{>0}$. Therefore
\[g^{\asc}_{-1,\mathrm{PFA}}(m,n)=\mathbb{N}_{>0}\] for $m,n \ge 1$. \qedbox
\end{proof}

\begin{corollary}
\label{cor:quotient-binary-spectrum}
For $m,n\ge 1$, the right-quotient accepting-state spectrum over any fixed alphabet of size at least two is $\mathbb{N}_{>0}$. The smallest nonempty alphabet on which a target value $\alpha\ge 1$ is attainable has size $1$ if $\alpha\le mn$ and size $2$ if $\alpha>mn$.
\end{corollary}

\begin{proof}
For $1\le\alpha\le mn$, take the unary witnesses from Proposition~\ref{prop:quotient-unary-range} and add $b$ as an identity letter to both automata. Erasing $b$, deleting every occurrence of it from a word input to the automata, preserves that each operand satisfies reachability and distinguishability, and thus this extension preserves the accepting-state complexities of the operands and their quotient.

For $\alpha>mn$, use the state sets, initial states, and final sets of the ternary witnesses $A^{\mathrm q}_{m,\alpha}$ and $B^{\mathrm q}_{n,\alpha}$. Let the alphabet consist of only the letters $a$ and $b$.  We can define the transition function such that $a$ acts as the word $ac$ in both automata, reserving the exact same role for $b$. $a$ therefore advances coordinate 2 of the right automaton by one modulo $n+1$.

On $[k]$, where $k=\alpha+1$, the word $ab$ fixes $1$ and cycles $2,\ldots,k$. Hence the binary word
\[
z=a^k(ab)^{n(k-1)}
\]
fixes the first coordinate and advances the second by $k+n(k-1)\equiv 1\pmod{n+1}$. Thus $z$ acts as the original $c$, and $az^n$ acts as the original $a$. The result is that the same pairs of transformations are realized by the witness pairs: in each operand, they have the same reachable states and pairs of distinguishable states, and the same final set for the quotient by Lemma~\ref{lem:quotient-final-set}. The left-hand actions also give the same reachability and distinguishability in the quotient. Thus the binary witnesses realize the same complexities $m,n,\alpha$.

The binary alphabet extends to any larger alphabet by letting each additional letter induce the identity function on the state set of each automaton. Proposition~\ref{prop:quotient-zero-impossible} excludes $0$, and the unary upper bound of Rauch and Holzer~\cite[Corollary~3.30]{RauchHolzer2023} rules out a one-letter witness when $\alpha>mn$.
\end{proof}

\bibliographystyle{splncs04}
\bibliography{references}

\end{document}